\documentclass[a4paper,12pt]{article}

\usepackage{color}
\usepackage{amssymb, amsmath}

\newcommand\EFFACE[1]{}

\newtheorem{theorem}{Theorem}
\newtheorem{proposition}[theorem]{Proposition}

\newtheorem{corollary}[theorem]{Corollary}
\newtheorem{conjecture}[theorem]{Conjecture}

\newenvironment{proof}{
\par
\noindent {\bf Proof.}\rm}{\mbox{}\hfill$\square$\par\vskip 3mm}

\def\NDI{{\rm ndi}}
\def\UND{{\rm und}}

\begin{document}

\title{{\bf A note on the neighbour-distinguishing\\ index of digraphs}}

\author{\'Eric Sopena~\thanks{Univ. Bordeaux, CNRS, Bordeaux INP, LaBRI, UMR5800, F-33400 Talence, France.}
\and
Mariusz Wo\'zniak~\thanks{AGH University of Science and Technology, al. A. Mickiewicza 30, 30-059 Krakow, Poland.}
}

\maketitle

\abstract{
In this note, we introduce and study a new version of neighbour-distinguishing arc-colourings of digraphs.
An arc-colouring $\gamma$ of a digraph $D$ is proper if no two arcs with the same head or
with the same tail are assigned the same colour. 
For each vertex $u$ of $D$, we denote by $S_\gamma^-(u)$ and $S_\gamma^+(u)$ the sets of colours that
appear on the incoming arcs and on the outgoing arcs of $u$, respectively.
An arc colouring $\gamma$ of $D$ is \emph{neighbour-distinguishing} 
if, for every two adjacent vertices $u$ and $v$ of $D$, 
the ordered pairs $(S_\gamma^-(u),S_\gamma^+(u))$ and $(S_\gamma^-(v),S_\gamma^+(v))$
are distinct.
The neighbour-distinguishing index of $D$ is then the smallest number of colours needed
for a neighbour-distinguishing arc-colouring of $D$.

We prove upper bounds on the neighbour-distingui\-shing index of various classes of digraphs.
}

\bigskip


\medskip


\section{Introduction}
\label{sec:intro}

A proper edge-colouring of a graph $G$ is \emph{vertex-distinguishing} 
if, for every two vertices $u$ and $v$ of $G$, the sets of colours that
appear on the edges incident with $u$ and $v$ are distinct.
Vertex-distinguishing proper edge-colourings of graphs were independently
introduced by 
Burris and Schelp~\cite{BS97}, 
and by \v Cerny, Hor\v n\'ak and Sot\'ak~\cite{CHS96}.
Requiring only adjacent vertices to be distinguished led to the notion
of \emph{neighbour-distinguishing} edge-colourings,
considered in~\cite{BGLS07,EHW06,ZLW02}.

Vertex-distinguishing arc-colourings of digraphs have been recently introduced
and studied by Li, Bai, He and Sun~\cite{LBHS16}.
An arc-colouring of a digraph is proper if no two arcs with the same head or
with the same tail are assigned the same colour.
Such an arc-colouring is \emph{vertex-distinguishing} 
if, for every two vertices $u$ and $v$ of $G$, 
(i)~the sets $S^-(u)$ and $S^-(v)$ of colours that
appear on the incoming arcs of $u$ and $v$, respectively, are distinct,
and (ii)~the sets $S^+(u)$ and $S^+(v)$ of colours that
appear on the outgoing arcs of $u$ and $v$, respectively, are distinct.

In this paper, we introduce and study a neighbour-distinguishing 
version of arc-colourings of digraphs,
using a slightly different distinction criteria:
two neighbours $u$ and $v$ are distinguished 
whenever $S^-(u)\neq S^-(v)$ or $S^+(u)\neq S^+(v)$.

Definitions and notation are introduced in the next section.
We prove a general upper bound on the neighbour-distinguishing
index of a digraph in Section~\ref{sec:general}, and study
various classes of digraphs in Section~\ref{sec:classes}.
Concluding remarks are given in Section~\ref{sec:discussion}.

\section{Definitions and notation}
\label{sec:definitions}

All digraphs we consider are without loops and multiple arcs.
For a digraph $D$, we denote by $V(D)$ and $A(D)$ its sets of vertices and arcs, respectively.
The {\em underlying graph} of $D$, denoted $\UND(D)$, is the simple undirected graph obtained from $D$
by replacing each arc $uv$ (or each pair of arcs $uv$, $vu$) by the edge $uv$.

If $uv$ is an arc of a digraph $D$, $u$ is the {\em tail} and $v$ is the {\em head} of~$uv$.
For every vertex $u$ of $D$, we denote by $N_D^+(u)$ and $N_D^-(u)$ the sets of {\em out-neighbours}
and {\em in-neighbours} of $u$, respectively.
%
Moreover, we denote by $d_D^+(u)=|N_D^+(u)|$ and $d_D^-(u)=|N_D^-(u)|$ the \emph{outdegree} and \emph{indegree} of $u$, 
respectively, and by $d_D(u)=d_D^+(u)+d_D^-(u)$ the \emph{degree} of $u$.

For a digraph $D$, we denote by 
$\delta^+(D)$, $\delta^-(D)$, $\Delta^+(D)$ and $\Delta^-(D)$
the minimum outdegree, minimum indegree, maximum outdegree and maximum indegree of $D$, respectively.
Moreover, we let
$$\Delta^*(D)=\max\{\Delta^+(D),\ \Delta^-(D)\}.$$


\medskip

A (proper) {\em $k$-arc-colouring} of a digraph $D$ is a mapping $\gamma$ from $V(D)$ 
to a set of $k$ colours (usually $\{1,\dots,k\}$) such that, for every vertex $u$, 
(i)~any two arcs with head $u$ are assigned distinct colours, and
(ii)~any two arcs with tail $u$ are assigned distinct colours.
Note here that two consecutive arcs $vu$ and $uw$, $v$ and $w$ not necessarily
distinct, may be assigned the same colour.
The {\em chromatic index} $\chi'(D)$ of a digraph $D$ is then the smallest number
$k$ for which $D$ admits a $k$-arc-colouring.

The following fact is well-known (see e.g.~\cite{LBHS16,W96,Z06}).

\begin{proposition}
For every digraph $D$, $\chi'(D)=\Delta^*(D)$.
\label{prop:arc-colouring}
\end{proposition}

For every vertex $u$ of a digraph $D$, and every arc-colouring $\gamma$ of $D$, 
we denote by $S_\gamma^+(u)$ and $S_\gamma^-(u)$
the sets of colours assigned by $\gamma$ to the outgoing and incoming arcs of $u$, respectively.
From the definition of an arc-colouring, we get $d_D^+(u)=|S_\gamma^+(u)|$ and $d_D^-(u)=|S_\gamma^-(u)|$ for
every vertex $u$.

We say that two vertices $u$ and $v$ of a digraph $D$
are {\em distinguished} by an arc-colouring $\gamma$ of $D$, if
$(S^+_\gamma(u),S^-_\gamma(u))\neq (S^+_\gamma(v),S^-_\gamma(v))$.
Note that we consider here ordered pairs, so that $(A,B)\neq (B,A)$ whenever $A\neq B$.
Note also that if $u$ and $v$ are such that $d_D^+(u)\neq d_D^+(v)$ or $d_D^-(u)\neq d_D^-(v)$,
which happens in particular if $d_D(u)\neq d_D(v)$,
then they are distinguished by every arc-colouring of $D$.
We will write $u\nsim_\gamma v$ if $u$ and $v$ are distinguished by $\gamma$
and $u\sim_\gamma v$ otherwise.



A $k$-arc-colouring $\gamma$ of a digraph $D$ is {\em neighbour-distinguishing} 
if $u\nsim_\gamma v$ for every arc $uv\in A(D)$. 
Such an arc-colouring will be called an {\em nd-arc-colouring} for short.
The {\em neighbour-distinguishing index} $\NDI(D)$ of a digraph $D$ is then the smallest number
of colours required for an nd-arc-colouring of $D$.

The following lower bound is easy to establish.

\begin{proposition}
For every digraph $D$, $\NDI(D)\geq\chi'(D)=\Delta^*(D)$.
Moreover, if there are two vertices $u$ and $v$ in $D$ 
with $d_D^+(u)=d_D^+(v)=d_D^-(u)=d_D^-(v)=\Delta^*(D)$,
then $\NDI(D)\geq\Delta^*(D)+1$.
\label{prop:lower-bound}
\end{proposition}

\begin{proof}
The first statement follows from the definitions.
For the second statement, observe that
$S_\gamma^+(u)=S_\gamma^+(v)=S_\gamma^-(u)=S_\gamma^-(v)=\{1,\dots,\Delta^*(D)\}$
for any two such vertices $u$ and $v$ and any $\Delta^*(D)$-arc-colouring
$\gamma$ of $D$.
\end{proof}

\section{A general upper bound}
\label{sec:general}

If $D$ is an \emph{oriented graph}, that is, a digraph with no opposite arcs,
then every proper edge-colouring $\varphi$ of $\UND(D)$ is an nd-arc-colouring of $D$ since,
for every arc $uv$ in $D$, $\varphi(uv)\in S_\varphi^+(u)$ and $\varphi(uv)\notin S_\varphi^+(v)$,
which implies $u\nsim_\varphi v$.
Hence, we get the following upper bound for oriented graphs, thanks to classical Vizing's bound.

\begin{proposition}\label{prop:oriented-graphs}
If $D$ is an oriented graph, then 
$$\NDI(D)\le\chi'(\UND(D))\le\Delta(\UND(D))+1\le 2\Delta^*(D)+2.$$
\end{proposition}

However, a proper edge-colouring of $\UND(D)$ may produce an arc-colouring of $D$
which is not neighbour-distinguishing when $D$ contains opposite arcs.
Consider for instance the digraph $D$ given by 
$V(D)=\{a,b,c,d\}$ and $A(D)=\{ab,bc,cb,dc\}$.
We then have $\UND(D)=P_4$, the path of order~4, and thus $\chi'(\UND(D))=2$.
It is then not difficult to check that for any 2-edge-colouring $\varphi$ of $\UND(D)$,
$S_\varphi^+(b)=S_\varphi^+(c)$ and $S_\varphi^-(b)=S_\varphi^-(c)$.

\medskip

We will prove that the upper bound given in Proposition~\ref{prop:oriented-graphs}
can be decreased to $2\Delta^*(D)$, even when $D$ contains opposite arcs.
Recall that a digraph $D$ is \emph{$k$-regular} if $d_D^+(v)=d_D^-(v)=k$ for every
vertex $v$ of~$D$.
A \emph{$k$-factor} in a digraph $D$ is a spanning $k$-regular subdigraph of~$D$.
The following result is folklore.

\begin{theorem}
Every $k$-regular digraph can be decomposed into $k$ arc-disjoint $1$-factors.
\label{th:k-regular-digraph}
\end{theorem}

We first determine the neighbour-distinguishing index of a 1-factor.

\begin{proposition}
If $D$ is a digraph with $d_D^+(u)=d_D^-(u)=1$ for every vertex $u$ of $D$,
then $\NDI(D)=2$.
\label{prop:1-factor}
\end{proposition}

\begin{proof}
Such a digraph $D$ is a disjoint union of directed cycles
and any such cycle needs at least two colours to be neighbour-distinguished.
An nd-arc-colouring of $D$ using two colours can be obtained as follows.
For a directed cycle of even length, use alternately colours $1$ and $2$.
For a directed cycle of odd length, use the colour $2$ on any two
consecutive arcs, and then use alternately colours $1$ and $2$.
The so-obtained 2-arc-colouring is clearly neighbour-distinguishing,
so that $\NDI(D)=2$.
\end{proof}

We are now able to prove the following general upper bound on the 
neighbour-distinguishing index of a digraph.

\begin{theorem}
For every digraph $D$, $\NDI(D)\le 2\Delta^*(D)$.
\label{th:general-bound}
\end{theorem}

\begin{proof}
Let $D'$ be any $\Delta^*(D)$-regular digraph containing $D$ as a subdigraph. 
If $D$ is not already regular, such a digraph can be obtained from $D$ by adding new arcs, 
and maybe new vertices.

By Theorem~\ref{th:k-regular-digraph}, the digraph $D'$ can be decomposed into
$\Delta^*(D')=\Delta^*(D)$ arc-disjoint 1-factors, say $F_1,\dots,F_{\Delta^*(D)}$.
By Proposition~\ref{prop:1-factor}, we know that $D'$ admits an nd-arc-colouring $\gamma'$
using $2\Delta^*(D')=2\Delta^*(D)$ colours.
We claim that the restriction $\gamma$ of $\gamma'$ to $A(D)$ is also neighbour-distinguishing.

To see that, let $uv$ be any arc of $D$,
and let $t$ and $w$ be the two vertices such that the directed walk $tuvw$
belongs to a 1-factor $F_i$ of $D'$ for some $i$, $1\le i\le \Delta^*(D)$.
Note here that we may have $t=w$, or $w=u$ and $t=v$.
If $\gamma(uv)\neq\gamma'(vw)$, then $\gamma(uv)\in S_\gamma^+(u)$ and $\gamma(uv)\notin S_\gamma^+(v)$.
Similarly,
if $\gamma'(tu)\neq\gamma(uv)$, then $\gamma(uv)\in S_\gamma^-(v)$ and $\gamma(uv)\notin S_\gamma^-(u)$.
Since neither three consecutive arcs nor two opposite arcs
in a walk of a 1-factor of $D'$ are assigned the same colour by $\gamma'$,
we get that $u\nsim_\gamma v$ for every arc $uv$ of $D$, as required.

This completes the proof. 
\end{proof}

\section{Neighbour-distinguishing index of some\\ classes of digraphs}
\label{sec:classes}

We study in this section the neighbour-distinguishing index of several
classes of digraphs, namely complete symmetric digraphs, bipartite digraphs
and digraphs whose underlying graph is $k$-chromatic, $k\ge 3$.

\subsection{Complete symmetric digraphs}
\label{subsec:complete}

We denote by $K_n^*$ the complete symmetric digraph of order $n$.
%
Observe first that any proper edge-colouring $\epsilon$ of $K_n$
induces an
arc-colouring $\gamma$ of $K_n^*$ defined by
$\gamma(uv)=\gamma(vu)=\epsilon(uv)$ for every edge $uv$ of $K_n$.
Moreover, since $S_\gamma^+(u)=S_\gamma^-(u)=S_\epsilon(u)$ for every
vertex $u$, $\gamma$ is neighbour-distinguishing whenever $\epsilon$
is neighbour-distinguishing.
Using a result of Zhang, Liu and Wang (see Theorem~6 in~\cite{ZLW02}),
we get that $\NDI(K_n^*)=\Delta^*(K_n^*)+1=n$ if $n$ is odd,
and $\NDI(K_n^*)\leq\Delta^*(K_n^*)+2=n+1$ if $n$ is even.

We prove that the bound in the even case can be decreased by one
(we recall the proof of the odd case to be complete).

\begin{theorem}
For every integer $n\ge 2$, $\NDI(K_n^*)=\Delta^*(K_n^*)+1= n$.
\end{theorem}

\begin{proof} 
Note first that we necessarily have $\NDI(K_n^*)\geq n$ for every $n\geq 2$ by 
Proposition~\ref{prop:lower-bound}.
Let $V(K_n^*)=\{v_0,\dots,v_{n-1}\}$. 
If $n=2$, we obviously have $\NDI(K_2^*)=|A(K_2^*)|= 2$ and the result follows.
We can thus assume $n\ge 3$.
We consider two cases, depending on the parity of $n$.

Suppose first that $n$ is odd, and consider a partition of the set
of edges of $K_n$ into $n$ disjoint maximal matchings, say $M_0,\dots,M_{n-1}$,
such that for each $i$, $0\le i\le n-1$, 
the matching $M_i$ does not cover the vertex~$v_i$.
We define an $n$-arc-colouring $\gamma$ of $K_n^*$ 
(using the set of colours $\{0,\dots,n-1\}$) as follows.
For every $i$ and~$j$, $0\le i<j\le n-1$, we set
$\gamma(v_iv_j)=\gamma(v_jv_i)=k$ if and only if the edge $v_iv_j$ belongs to $M_k$.
Observe now that for every vertex $v_i$, $0\le i\le n-1$, the colour $i$ is the unique colour
that does not belong to $S^+_\gamma(v_i)\cup S^-_\gamma(v_i)$, since $v_i$ is not covered
by the matching $M_i$. 
This implies that $\gamma$ is an nd-arc-colouring of $K_n^*$,
and thus $\NDI(K_n^*) = n$, as required.

Suppose now that $n$ is even.
Let $K'$ be the subgraph of $K_n^*$ induced by the set of vertices $\{v_0,\dots,v_{n-2}\}$
and $\gamma'$ be the $(n-1)$-arc-colouring of $K'$ defined as above.
We define an $n$-arc-colouring $\gamma$ of $K_n^*$ (using the set of colours $\{0,\dots,n-1\}$)
as follows:
\begin{enumerate}
\item 
for every $i$ and~$j$, $0\le i<j\le n-2$, $j\not\equiv i+1\pmod{n-1}$, we set $\gamma(v_iv_j)=\gamma'(v_iv_j)$,
\item 
for every $i$, $0\le i\le n-2$, we set $\gamma(v_iv_{i+1})=n-1$ and $\gamma(v_{i+1}v_i)=\gamma'(v_{i+1}v_i)$
(subscripts are taken modulo $n-1$),
\item 
for every $i$, $0\le i\le n-2$,
we set $\gamma(v_{n-1}v_i)=\gamma'(v_{i-1}v_i)$ and $\gamma(v_iv_{n-1})=\gamma'(v_{i+1}v_i)$.
\end{enumerate}
Since the colour $n$ belongs to $S^+_\gamma(v_i)\cap S^-_\gamma(v_i)$
for every $i$, $0\le i\le n-2$, and does not belong to $S^+_\gamma(v_{n-1})\cup S^-_\gamma(v_{n-1})$,
the vertex $v_{n-1}$ is distinguished from every other vertex in $K_n^*$.
Moreover, for every vertex $v_i$, $0\le i\le n-2$,
$$S_\gamma^+(v_i) = S_{\gamma'}^+(v_i)\cup\{n-1\}
\ \ \mbox{and}\ \ S_\gamma^-(v_i) = S_{\gamma'}^-(v_i)\cup\{n-1\},$$
which implies that any two vertices $v_i$ and $v_j$, $0\le i<j\le n-2$, 
are distinguished since $\gamma'$ is an nd-arc-colouring of $K'$.
%
We thus get that $\gamma$ is an nd-arc-colouring of $K_n^*$,
and thus $\NDI(K_n^*)\leq n$, as required.

This completes the proof.
\end{proof}

\subsection{Bipartite digraphs}
\label{subsec:bipartite}

A digraph $D$ is {\em bipartite} if its underlying graph
is bipartite. In that case, $V(D)=X\cup Y$ with $X\cap Y=\emptyset$ and
$A(D)\subseteq X\times Y\cup Y\times X$.
We then have the following result.

\begin{theorem}
If $D$ is a bipartite digraph, then $\NDI(D)\le\Delta^*(D)+2$.
\label{th:bipartite}
\end{theorem}

\begin{proof}
Let $V(D)=X\cup Y$ 
be the bipartition of $V(D)$ and $\gamma$ be any (not necessarily
neighbour-distinguishing) optimal arc-colouring of $D$
using $\Delta^*(D)$ colours (such an arc-colouring exists by Proposition~\ref{prop:arc-colouring}).

If $\gamma$ is an nd-arc-colouring we are done.
Otherwise, let $M_1\subseteq A(D)\cap(X\times Y)$ be a maximal matching from $X$
to $Y$.
We define the arc-colouring $\gamma_1$ as follows:
\begin{center}$\gamma_1(uv)=\Delta^*(D)+1$ if $uv\in M_1$, $\gamma_1(uv)=\gamma(uv)$ otherwise.\end{center}   
Note that if $uv$ is an arc such that $u$ or $v$ is (or both are) covered by $M_1$, then
$u\nsim_{\gamma_1} v$ since the colour $\Delta^*(D)+1$ appears in exactly
one of the sets $S_{\gamma_1}^+(u)$ and $S_{\gamma_1}^+(v)$, or
in exactly one of the sets $S_{\gamma_1}^-(u)$ and $S_{\gamma_1}^-(v)$.

If $\gamma_1$ is an nd-arc-colouring we are done.
Otherwise, let $A^\sim$ be the set of arcs $uv\in A(D)$ with $u\sim_{\gamma_1} v$
and $M_2\subseteq A^\sim\cap(Y\times X)$ be a maximal matching from $Y$
to $X$ of $A^\sim$.
We define the arc-colouring $\gamma_2$ as follows:
\begin{center}$\gamma_2(uv)=\Delta^*(D)+2$ if $uv\in M_2$, $\gamma_2(uv)=\gamma_1(uv)$ otherwise.\end{center}   
Again, note that if $uv$ is an arc such that $u$ or $v$ is (or both are) covered by $M_2$, then
$u\nsim_{\gamma_2} v$. 
Moreover, since $M_2$ is a matching of $A^\sim$, pairs of vertices that
were distinguished by $\gamma_1$ are still distinguished by $\gamma_2$.

Hence, every arc $uv$ such that $u$ and $v$ were not distinguished by $\gamma_1$
are now distinguished by $\gamma_2$ which is thus an nd-arc-colouring of $D$ using
$\Delta^*(D)+2$ colours. 
This concludes the proof. 
\end{proof}


The upper bound given in Theorem~\ref{th:bipartite} can be decreased when the underlying graph
of $D$ is a tree.

\begin{theorem}
If $D$ is a digraph whose underlying graph is a tree, then $\NDI(D)\le\Delta^*(D)+1$.
\label{th:tree}
\end{theorem}

\begin{proof}
The proof is by induction on the order $n$ of $D$.
The result clearly holds if $n\le 2$.
Let now $D$ be a digraph of order $n\ge 3$, such that the underlying graph $\UND(D)$ of $D$ is a tree,
and $P=v_1\dots v_k$, $k\le n$, be a path in $\UND(D)$ with maximal length.
By the induction hypothesis, there exists an nd-arc-colouring $\gamma$ of $D-v_k$ using at most
$\Delta^*(D-v_k)+1$ colours. We will extend $\gamma$ to an nd-arc-colouring of $D$ using at most
$\Delta^*(D)+1$ colours.

If $\Delta^*(D)=\Delta^*(D-v_k)+1$, we assign the new colour $\Delta^*(D)+1$ to the at most
two arcs incident with $v_k$ so that the so-obtained arc-colouring is clearly neighbour-distinguishing.

Suppose now that $\Delta^*(D)=\Delta^*(D-v_k)$. 
If all neighbours of $v_{k-1}$ are leaves, the underlying graph of $D$
is a star.
In that case, there is at most one arc linking $v_{k-1}$ and $v_k$,
and colouring this arc with any admissible colour produces an nd-arc-colouring
of $D$.
If the underlying graph of $D$ is not a star, then,
by the maximality of~$P$, we get that $v_{k-1}$ has exactly one neighbour
which is not a leaf, namely~$v_{k-2}$.
This implies that the only conflict that might appear when colouring the arcs
linking $v_k$ and $v_{k-1}$ is between $v_{k-2}$ and $v_{k-1}$
(recall that two neighbours with distinct indegree or outdegree are necessarily distinguished).

Since $d_D^+(v_{k-2})\le\Delta^*(D)$ and $d_D^-(v_{k-2})\le\Delta^*(D)$,
there necessarily exist a colour $a$ such that $S_\gamma^+(v_{k-2})\neq S_\gamma^+(v_{k-1})\cup\{a\}$,
and a colour $b$ such that $S_\gamma^-(v_{k-2})\neq S_\gamma^-(v_{k-1})\cup\{b\}$.
Therefore, the at most two arcs incident with $v_k$ can be coloured, using $a$ and/or $b$,
in such a way that the so-obtained arc-colouring is neighbour-distinguishing.

This completes the proof.
\end{proof}

\subsection{Digraphs whose underlying graph is $k$-chromatic}
\label{subsec:chromatic}

Since the set of edges of every $k$-colourable graph can be partitionned in $\lceil \log k\rceil$
parts each inducing a bipartite graph (see e.g. Lemma 4.1 in~\cite{BGLS07}), Theorem~\ref{th:bipartite} leads to the
following general upper bound:

\begin{corollary}
If $D$ is a digraph whose underlying graph has chromatic number $k\ge 3$, then
$\NDI(D)\le\Delta^*(D)+2\lceil \log k\rceil$.
\label{cor:general-bound}
\end{corollary}

\begin{proof}
Starting from an optimal arc-colouring of $D$ with $\Delta^*(D)$ colours, it suffices to use two new colours for each 
of the $\lceil \log k\rceil$ bipartite parts (obtained from any optimal vertex-colouring of the
underlying graph of $D$), as shown in the proof of Theorem~\ref{th:bipartite}, in order
to get an nd-arc-colouring of $D$.
\end{proof}

\section{Discussion}
\label{sec:discussion}

In this note, we have introduced and studied a new version of neighbour-distingui\-shing 
arc-colourings of digraphs. 
Pursuing this line of research, we propose the following questions. 

\begin{enumerate}
\item Is there any general upper bound on the neighbour-distinguishing index of symmetric digraphs?
\item Is there any general upper bound on the neighbour-distinguishing index of 
not necessarily symmetric complete digraphs?
\item Is there any general upper bound on the neighbour-distinguishing index of directed acyclic graphs?
\item The general bound given in Corollary~\ref{cor:general-bound} is certainly
not optimal. 
In particular,
is it possible to improve this bound for digraphs whose underlying graph is 3-colourable?
\end{enumerate}

We finally propose the following conjecture.

\begin{conjecture}
For every digraph $D$, $\NDI(D)\le\Delta^*(D)+1$.
\end{conjecture}


\begin{thebibliography}{99}


\bibitem{BGLS07}
P. N. Balister, E. Gy\"ori, J. Lehel, R.H. Schelp.
Adjacent vertex distinguishing edge-colorings.
{\it SIAM J. Discrete Math.} 21 (2007), 237--250.

\bibitem{BS97}
A.C. Burris, R.H. Schelp. 
Vertex-distinguishing proper edge-colorings.
{\it J. Graph Theory} 26 (1997), 73--82.

\bibitem{CHS96}
J. \v Cerny, M. Hor\v n\'ak, R. Sot\'ak, 
Observability of a graph.
{\it Math. Slovaca} 46(1) (1996), 21--31. 
 
\bibitem{EHW06}
K. Edwards, M. Hor\v n\'ak, M. Wo\'zniak.
On the Neighbour-Distinguishing Index of a Graph.
{\it Graphs Combin.} 22 (2006), 341--350.


\bibitem{LBHS16}
H. Li, Y. Bai, W. He, Q. Sun.
Vertex-distinguishing proper arc colorings of digraphs.
{\it Discrete Applied Math.} 209 (2016), 276--286.

\bibitem{W96}
D.B. West.
{\it Introduction to Graph Theory}.
Prentice Hall, N.~J. (1996).

\bibitem{ZLW02}
Z. Zhang, L. Liu, J. Wang.
Adjacent strong edge coloring of graphs.
{\it Appl. Math. Lett.} 15 (2002), 623--626.


\bibitem{Z06}
M. Zwonek.
On arc coloring of digraphs.
{\it Opuscula Math.} 26 (2006), 185--195.

\end{thebibliography}
\end{document}